\documentclass[10pt,twocolumn,twoside]{ieeeconf}
\IEEEoverridecommandlockouts    
\usepackage{amsthm}
\usepackage{amsmath}
\usepackage{amssymb}
\usepackage{mathtools}
\usepackage{dsfont}
\usepackage{hyperref}
\let\labelindent\relax
\usepackage{enumitem}
\usepackage{color}
\usepackage{graphicx}
\usepackage{multirow}
\usepackage{geometry}
\newtheorem{corollary}{\normalfont \textbf{Corollary}} 
\newtheorem{assumption}{\normalfont \textbf{Assumption}}
\newtheorem{lemma}{\normalfont \textbf{Lemma}}
\newtheorem{theorem}{\normalfont \textbf{Theorem}}
\newtheorem{proposition}{\normalfont \textbf{Proposition}}
\newtheorem{definition}{\normalfont \textbf{Definition}}
\newtheorem{remark}{\normalfont \textbf{Remark}}

\usepackage{caption}
\usepackage{subcaption}
\usepackage{cleveref}
\usepackage{cite}

\title{\textbf{Robust Safety-Critical Control of Networked SIR Dynamics}}
\author{Saba Samadi, Brooks A. Butler, and Philip E. Par\'e* \thanks{*Saba Samadi and Philip E.~Par\'e are with the Elmore Family School of Electrical and Computer Engineering at Purdue University. Brooks A. Butler is with the Department of Electrical Engineering and Computer Science at the University of California, Irvine. Emails:~ssamadi@purdue.edu, philpare@purdue.edu, bbutler2@uci.edu. This work was supported in part 
   by the National Science Foundation, grants
   NSF-ECCS \#2238388 and by an appointment to the Intelligence Community Postdoctoral Research Fellowship Program at the University of California, Irvine administered by Oak Ridge Institute for Science and Education (ORISE) through an interagency agreement between the U.S. Department of Energy and the Office of the Director of National Intelligence (ODNI).}}

\begin{document}

\maketitle

\begin{abstract}
We present a robust safety-critical control framework tailored for networked susceptible-infected-recovered (SIR) epidemic dynamics, leveraging control barrier functions (CBFs) and robust control barrier functions to address the challenges of epidemic spread and mitigation. In our networked SIR model, each node must keep its infection level below a critical threshold, despite dynamic interactions with neighboring nodes and inherent uncertainties in the epidemic parameters and measurement errors, to ensure public health safety. We first derive a CBF-based controller that guarantees infection thresholds are not exceeded in the nominal case. We enhance the framework to handle realistic epidemic scenarios under uncertainties by incorporating compensation terms that reinforce safety against uncertainties: an independent method with constant bounds for uniform uncertainty, and a novel approach that scales with the state to capture increased relative noise in early or suppressed outbreak stages. Simulation results on a networked SIR system illustrate that the nominal CBF controller maintains safety under low uncertainty, while the robust approaches provide formal safety guarantees under higher uncertainties; in particular, the novel method employs more conservative control efforts to provide larger safety margins, whereas the independent approach optimizes resource allocation by allowing infection levels to approach the boundaries in steady epidemic regimes.
\end{abstract}

\section{Introduction}
As global mobility increases and interconnectedness among major population centers grows, infectious diseases have the potential to spread rapidly across regions, heightening the risk of epidemic outbreaks. Consequently, developing sophisticated and explainable epidemic control strategies is a critical area of public health research \cite{anderson1992infectious,andersson2012stochastic,keeling2008modeling,sereno2022minimizing}, as decision-makers face significant complexity and uncertainty at each step of the epidemic estimation and mitigation process.  
In this work, we consider the control of a networked 
susceptible-infected-recovered (SIR) epidemic process \cite{pastor2015epidemic,pare2020modeling}, which models the spread of diseases where infected members of interconnected populations cannot be reinfected with the same disease after they have recovered from the infection. In such a system, we aim to provide insights into how decision-makers at each node (i.e., population center) can ensure that the number of infections does not exceed a specified limit while also accounting for the dynamic influence of neighboring nodes with varying levels of uncertainty. In this context, the application of control barrier functions (CBFs) \cite{ames2016control, ames2019control} can provide a useful tool for guaranteeing population safety, as defined by a barrier function. This is particularly relevant, since epidemic control often focuses on preventing unsafe infection levels that could overwhelm healthcare/treatment facilities. 

While some initial work exists that demonstrates the usefulness of CBFs towards enforcing safety conditions in the event of an epidemic outbreak \cite{ames2020safety,molnar2020safety,butler2023optimal}, the application of CBFs in dynamically coupled nonlinear networked systems is a challenging and open problem in the field of safety-critical control \cite{butler2024resilience,butler2023collaborative}. In the case of networked epidemics, we are especially interested in the incorporation of robustness into the decision-making for epidemic mitigation, where high degrees of uncertainty in the estimation of the true number of infected individuals, as well as uncertainty in the effectiveness of mitigation policies meant to reduce infection rates or increase recovery rates \cite{watkins2019robust,sharifi2017nonlinear,small2020modelling,ibeas2014robust,badfar2022design,bloem2009optimal,azimi2022state}. As a result of such uncertainty, decision-makers may choose to be either overly cautious (at a high economic cost) or too relaxed (at a high public health cost) in setting policies for mitigation efforts. In this regard, we leverage tools from robust safety-critical control \cite{alan2025generalizing}, which provide an explainable framework for selecting the least-costly action that still guarantees safety under uncertainty, and we extend this approach by developing a new low-prevalence–amplified compensation tailored for epidemics, designed to preserve safety when early-stage case counts are difficult to measure.

Contrasting with previous work, we develop a framework for safety-critical control of networked SIR epidemic processes with the following contributions:
\begin{enumerate}
    \item We provide a closed-form, node-level CBF controller that ensures the safety of each node in the SIR epidemic process with respect to the safety conditions of neighboring nodes.
    
    \item We extend this framework to a robust CBF
(RCBF) formulation tailored to epidemic control that addresses uncertainty while preserving formal safety guarantees.

    \item To improve robustness in the early stages of an epidemic, we propose a novel uncertainty model that captures amplified noise at low infection levels to design a new RCBF framework. 
\end{enumerate}

The remainder of this paper is structured as follows. The dynamics of the networked SIR model are introduced in Section~\ref{sec2}. The problem formulation is presented in Section~\ref{sec3}. Section~\ref{sec4} examines the nominal CBF design. Section~\ref{sec: robust} extends the framework to robust CBFs by incorporating uncertainty through compensation terms. Section~\ref{secres} illustrates the effectiveness of the proposed controllers with simulations on an epidemic network. Finally, Section~\ref{secfin} concludes the paper and discusses future research
directions.

\section{Networked SIR Model Dynamics}\label{sec2}
	
	In this section, we introduce the dynamics of the networked SIR model as follows:	
	\begin{subequations}\label{eq:SIR}
    \begin{align}
	\dot{x}_i &= -(\gamma_i + u_i) x_i + s_i \sum_{j=1}^n \beta_{ij} x_j \label{eq:SIR_x},\\
	\dot{r}_i &= (\gamma_i + u_i) x_i \label{eq:SIR_r},
    \end{align}
    \end{subequations}	
	where \(n\) represents the number of nodes in the network; \( x_i, r_i \in [0,1] \) are the fractions of infected and recovered individuals at node \( i \), respectively; \( s_i = 1 - x_i - r_i \) is the fraction of susceptible individuals at node \( i, \) which ensures that the total population at node \(i\) remains constant/well-defined; the recovery (removed) rate of each node is represented as \( \gamma_i > 0 \); \( \beta_{ij} \geq 0 \) is the infection rate from node \( j \) to node \( i \); and the control input at node \( i \) is denoted as \( u_i \in U_i = [0, \bar{u}_i] \subseteq \mathbb{R}_{\geq 0} \), where \(\bar{u}_i\) represents the maximum available control effort. The control input \(u_i\) is interpreted as an intervention effort aimed at increasing the healing rate at node \(i\); such boosting can be achieved through measures like enhanced medical treatment, enforced quarantine, rapid testing and isolation, or accelerated vaccination programs.
    The coupling between nodes, through the term \(\sum_{j=1}^n \beta_{ij} x_j\), reflects the influence of neighboring populations on the infection dynamics of node \(i\). 
    To express the system in a control-affine form, the dynamics can be rewritten as
    \begin{subequations}
        \begin{align}
        \begin{pmatrix}\dot{x}_i \\ \dot{r}_i \end{pmatrix} &= f_i(\mathbf{x}_i,r_i) + g_i(x_i) u_i, \nonumber \\
	f_i(\mathbf{x}_i,r_i) &= \begin{pmatrix}
		-\gamma_i x_i + (1 - x_i - r_i)\sum_{j=1}^n \beta_{ij} x_j \\[6pt]
		\gamma_i x_i
	\end{pmatrix}, \\
	g_i(x_i) 
        &= \begin{pmatrix} -x_i \\ x_i \end{pmatrix},
        \end{align}
        \end{subequations}
where \(\mathbf{x}_i = (x_i, x_{\mathcal{N}_i^+})\) denotes the one-hop neighborhood state vectors, \(\mathcal{N}_i^+\) is the set of nodes that directly influence node \(i\), and \(x_{\mathcal{N}_i^+}\) represents the fraction of infected individuals among these one-hop neighbors. 

\section{Problem Statement} \label{sec3}
	
To ensure that the infection level at each node remains below a critical threshold, we introduce a safety constraint. We define the safety constraint for each node as:
	\begin{equation}
        \label{eq.hi}
		h_i(x_i) = \bar{x}_i - x_i, 
	\end{equation}	
	where \( \bar{x}_i \in (0,1] \) is the maximum allowable infection level at node \( i \). We denote the time derivative of the safety function \(h_i(x_i)\) as \(\dot{h}_i(\mathbf{x}_i, r_i,u_i)\), which depends on both \(x_i\) and \(r_i\). The safety set for node \( i \) is	
	\begin{equation}
    \label{safeset}
		C_i = \{ x_i \in [0,1] : h_i(x_i) \geq 0 \}, 
	\end{equation}	
    which requires \(x_i(t) \leq \bar{x}_i\) for all time, in order for node $i$ to be considered safe. The safety condition may be interpreted as the maximum capacity of a node to treat infected individuals as constrained by limited medical resources.
\begin{remark}\label{rem:vanishing}
At the disease‐free equilibrium \(x_i=0\), the control input \( u_i \) has no effect on the dynamics, as \( g_i(x_i) = 0 \).
As $x_i$ approaches 0 from above, the system gradually loses controllability, with the magnitude of $g_i(x_i)$ diminishing proportionally to $x_i$, reducing the control authority.
However, since \( h_i(x_i) = \bar{x}_i > 0~\eqref{eq.hi},\)
the state is already inside the safe set and no intervention is required \cite{xi0}. Therefore, in the remainder of this paper, we focus on the non-trivial regime \(x_i \in (0,1]\).
\end{remark}
\begin{definition}
\label{dfn:safe}
A node \(i\) is said to be \emph{safe} if there exists an admissible control input \(u_i \in U_i\) such that the corresponding state trajectory \(x_i(t)\) remains in the safe set in \eqref{safeset} for all \(t \geq 0\). Equivalently, node \(i\) is safe if the control input can ensure that

\vspace{-3ex}

\begin{equation}
\label{safety}
    \dot{h}_i(\mathbf{x}_i, r_i, u_i) \geq -\alpha_i(h_i(x_i)),
\end{equation}
where \( \alpha_i \) is an extended class-\( \mathcal{K} \) function \cite{ames2016control}.
\end{definition} 

\begin{definition} \label{vulnerable}
For each node~$i$, define $\mathcal{V}_i:= \{\mathbf{x}_i :  h_j(x_j)\geq 0, \forall j\in \mathcal{N}_i^+\},$ the set of one-hop neighborhood states in which all neighbors of~$i$ satisfy their safety constraints.
We say that node~$i$ is \emph{vulnerable} if there exists a state $\mathbf{x}_i\in \mathcal{V}_i$ such that there is no input $u_i \in U_i$ that can satisfy \eqref{safety}.
\end{definition}

\textbf{Problem Formulation:}
Consider a network of interconnected populations modeled by the SIR dynamics in \eqref{eq:SIR}. Each node~$i$ must keep its infection level \(x_i\) below a safety threshold \(\bar{x}_i\) using a local control input \(u_i \in U_i\). The safety of node \(i\) depends on both its own state and the infection levels of its neighbors due to the non-linear infection term.
Given the infection tolerances (predefined safety boundaries) of node~$i$’s neighbors, we aim to determine the minimum required control effort $\bar{u}_i$ at node $i$. The goal is to ensure that the node is \emph{not vulnerable} as defined in Definition~\ref{vulnerable}. In particular, the control effort must guarantee that the time derivative of $h_i(x_i)$ is non-negative at the safety boundary $x_i = \bar{x}_i$. This requirement must hold even when the infection levels of neighboring nodes are at their maximum values, which exert the strongest negative influence on node~$i$’s safety. Addressing this worst-case scenario motivates the need for a robust mitigation strategy.

\section{Nominal Safety‑Critical Control} \label{sec4}
To maintain safety, we adopt a CBF framework \cite{ames2016control,ames2019control}. In this section, we first derive the CBF conditions that characterize the admissible control inputs, and then analyze node-level vulnerability. Following Remark~\ref{rem:vanishing}, we analyze the CBF condition on \(x_i\in(0,\bar{x}_i]\).

\subsection{CBF Derivation}
    We utilize the Lie derivative to calculate the first derivative of \( h_i(x_i) \) as follows
	\begin{equation}
		\dot{h}_i(\mathbf{x}_i, r_i, u_i) =  \mathcal{L}_{f_i} h_i(\mathbf{x}_i) + \mathcal{L}_{g_i} h_i(x_i) u_i,
	\end{equation}
    where the Lie derivative of a smooth function \(h: \mathbb{R}^n\to\mathbb{R}\) along a vector field \(f:\mathbb{R}^n\to\mathbb{R}^n\) at \(z\) is defined as 
    \[\mathcal{L}_f h(z) \triangleq \frac{\partial h(z)}{\partial z} f(z).\]
The calculation of \(\mathcal{L}_{f_i} h_i(\mathbf{x}_i)\) and \(\mathcal{L}_{g_i} h_i(x_i)\) referring to \eqref{eq.hi} is given by:
	\[
	\frac{\partial h_i}{\partial x_i} = -1, \quad \frac{\partial h_i}{\partial r_i} = 0.
	\]	
	Therefore:
	\begin{align}
	\mathcal{L}_{f_i} h_i &= \gamma_i x_i - (1 - x_i - r_i)\sum_j \beta_{ij} x_j, \nonumber \\
		\mathcal{L}_{g_i}h_i &= x_i, \text{and}\nonumber \\
	    \dot{h}_i(\mathbf{x}_i, r_i, u_i) &= (\gamma_i + u_i) x_i - (1 - x_i - r_i)\sum_{j=1}^n \beta_{ij} x_j.  \nonumber
        \end{align}

We aim to find the minimum required \( u_i \in U_i = [0, \bar{u}_i]\), such that \(u_i\) satisfies \eqref{safety}; therefore, the control input must be chosen as		
	\[
	(\gamma_i + u_i) x_i - (1 - x_i - r_i) \sum_{j=1}^n \beta_{ij} x_j + \alpha_i(\bar{x}_i - x_i) \geq 0,
	\]	
\begin{equation}
\label{safetycondition}
u_i \geq \frac{(1 - x_i - r_i) \sum_{j=1}^n \beta_{ij} x_j - \gamma_i x_i - \alpha_i(\bar{x}_i - x_i)}{x_i},
\end{equation}

\vspace{-3ex}

\small
\begin{align}
&u_i = \min\Biggl( \bar{u}_i,\, \label{finalu}\\
&\max\Bigl( 0,\, \frac{(1 - x_i - r_i)\sum_{j=1}^n \beta_{ij} x_j - \gamma_i x_i - \alpha_i(\bar{x}_i - x_i)}{x_i} \Bigr) \Biggr).\nonumber
\end{align}  

\normalsize

\noindent

Inequality $u_i \le \bar{u}_i$ is necessary for the control input $u_i$ to be feasible, which guarantees $u_i \in U_i$. Interpreting $u_i$ as the effort needed to enforce safety allows us to assess whether the resources capped by $\bar{u}_i$ are sufficient to keep the node safe. When the required effort exceeds $\bar{u}_i$, the node is vulnerable as defined in Section~\ref{sub.vul}, consistent with the bound in Corollary~\ref{cor:feas}; this perspective motivates the vulnerability analysis that follows.


\subsection{Node Vulnerability Analysis}\label{sub.vul}

We investigate the conditions under which a node becomes vulnerable by adopting a boundary-based approach. Specifically, we focus on the scenario when the node’s infection level reaches its upper threshold, i.e., when \( x_i = \bar{x}_i \) (or equivalently, \( h_i(x_i) = 0 \)~\eqref{eq.hi}). Therefore, at the boundary, ensuring safety requires that the time derivative \( \dot{h}_i(\mathbf{x}_i,\bar{x}_i,r_i, u_i) \) remains nonnegative,
	\begin{equation}
		\dot{h}_i(\mathbf{x}_i, \bar{x}_i, r_i, u_i) \geq - \alpha_i(0) = 0. \nonumber
	\end{equation}	
where \(\dot{h}_i(\mathbf{x}_i,\bar{x}_i,r_i,u_i)\) means evaluating \(\dot{h}_i(\mathbf{x}_i,r_i,u_i)\) at the boundary \(x_i=\bar{x}_i\).
	
\begin{definition} \label{viability}
For each node~\(i\), define 
\[
\mathcal{S}_i := \Bigl\{ \mathbf{x}_i : h_i(x_i)\geq 0 \text{ and } h_j(x_j)\geq 0, \ \forall j\in \mathcal{N}_i^+ \Bigr\},
\]
which is the set of one-hop neighborhood states in which node \(i\) and all its directly influencing neighbors are safe. We say that \(\mathcal{S}_i\) is a \emph{node-level viability domain} if, for every \(\mathbf{x}_i \in \mathcal{S}_i\), there exists an admissible control input \(u_i \in U_i\) that satisfies the safety condition \eqref{safety}, i.e., the set is forward invariant.
We denote \( \partial S_{i}\) as the set of points in the viability domain \(\mathcal{S}_i\) where \( h_i(x_i) = 0\) and \( h_j(x_j) = 0\).
\end{definition}

\begin{definition}
\label{NCBF SIR}
    Given the networked SIR dynamics in \eqref{eq:SIR} and the safety constraint \eqref{eq.hi}, a function \( h_i(x_i) \) is called a \emph{node-level control barrier function (NCBF)} for node \( i \) if, for all \( \mathbf{x}_i\) on the boundary of the node-level viability domain \( \partial S_{i}\), there exists a control input \(u_i\) such that \eqref{safety} holds.
    
\end{definition}
 NCBF ensures that the infected population of node~\( i \) remains below the threshold \( \bar{x}_i \) by appropriately choosing~\( u_i \), while considering the influence of its neighbors.

\begin{assumption}
\label{assum:u_bar}
For each node \(i\), the maximum control effort \(\bar{u}_i\) is sufficiently large so that for every infection level \(x_i \in (0,\bar{x}_i]\) and for all admissible recovery values \( r_i \in [0, 1 - x_i] \), and all possible infection levels of neighboring nodes $x_j \in (0,\bar{x}_j] $, there exists at least one control input \(u_i \in U_i\) that can satisfy the safety condition in 
\eqref{safetycondition}.
\end{assumption}
Having established that each node has a sufficiently large control input to meet its local safety requirement, we consider conditions to guarantee that safety is maintained throughout the network.
In particular, we require each node’s neighbors also to possess suitable CBFs, ensuring that the local safety analysis scales effectively to the entire network.
\begin{assumption}
\label{existanceNCBF}
Let $n_j = 1 + |\mathcal{N}_j^+| $, where $ \mathcal{N}_j^+ $ denotes the set of nodes that directly influence node $ j $. For every neighbor \(j\in \mathcal{N}_i^+\), there exists a node-level control barrier function \(h_j: \mathbb{R}^{n_j} \rightarrow \mathbb{R}\) with a corresponding viability domain \(S_j\), as defined in Definition \ref{viability}.   
\end{assumption}
Note that Assumption~\ref{existanceNCBF} is a best-case scenario for neighbors’ capabilities; we address uncertainties affecting neighbors in the next section.
With these assumptions, we can now formally state the conditions under which the safety set for each node is forward invariant. 
\begin{theorem}
    \label{thm:forwardInvariance}
    \textbf{Forward Invariance via NCBF for SIR Model} \\
    Given a networked dynamic SIR system defined by \eqref{eq:SIR} and under Assumptions~\ref{assum:u_bar} and~\ref{existanceNCBF}, the set \( S_i \subseteq C_i \) is forward invariant if and only if \( h_i \) is an NCBF.
\end{theorem}
\begin{proof}
\textbf{Necessity:} 

 Since \( h_i \) is an NCBF, by Definition \ref{NCBF SIR}, for all \( \mathbf{x}_i \in \partial S_{i} \), there exists \( u_i \in U_i \) such that:
        \[
        \dot{h}_i(\mathbf{x}_i, r_i, u_i) \geq -\alpha_i(h_i(x_i)).
        \]
        Since \( h_i(x_i) = \bar{x}_i - x_i \), and at the boundary \(x_i = \bar{x}_i\), we have \( h_i(x_i) = 0 \):
        \[
        \dot{h}_i(\mathbf{x}_i, \bar{x}_i, r_i, u_i) \geq 0.
        \]    
To ensure that node~\(i\) remains safe even under the most adverse conditions, we consider each neighbor \(j \in \mathcal{N}_i^+\) is at its infection threshold (\(x_j = \bar{x}_j\)). Since $h_i$ is an NCBF, under this worst-case scenario, and given Assumptions~\ref{assum:u_bar} and~\ref{existanceNCBF}, there must exist a control input \( u_i \in U_i \) such that:
\[
\begin{aligned}
\dot{h}_i(\mathbf{x}_i, \bar{x}_i, r_i, u_i) &= (\gamma_i + u_i)\bar{x}_i \\
&- \left(1 - \bar{x}_i - r_i\right)\left(\beta_{ii}\bar{x}_i + \sum_{j \in \mathcal{N}_i^+} \beta_{ij}\bar{x}_j\right) \geq 0.
\end{aligned}
\]
Hence, there is a feasible \(u_i\) (possibly \(\bar{u}_i\)) maintaining \(\dot{h}_i\ge0\). Dividing the inequality by \(\bar{x}_i \) yields the following equivalent form.
By choosing \(u_i\) at its upper bound \(\bar{u}_i\), we obtain a conservative condition that guarantees \(\dot{h}_i\ge0\):
    \begin{align*}  
		\gamma_i + \bar{u}_i - (1 - \bar{x}_i - r_i) \beta_{ii} - (1 - \bar{x}_i - r_i) \sum_{j \in \mathcal{N}_i^+} \beta_{ij} \frac{\bar{x}_j}{\bar{x}_i} \geq 0,	
    \end{align*}    
which ensures that \( h_i(x_i(t)) \) does not decrease below zero, maintaining \( x_i(t) \leq \bar{x}_i \).
       Therefore, starting from any state \(x_i(0)\in S_i\), the system cannot exit \( S_i \).
    Thus, 
    $
    S_i \subseteq C_i$ is forward invariant.

    \textbf{Sufficiency:} 
    
 Given \( S_i \subseteq C_i \) is forward invariant. Assume, by way of contradiction, that \( h_i \) is not an NCBF.
 Therefore, there exists some \( \mathbf{x}_i \in \partial S_{i} \) such that, for all \( u_i \in U_i = [0, \bar{u}_i] \):
        \[
        \dot{h}_i(\mathbf{x}_i, r_i, u_i) < -\alpha_i(h_i(x_i)) = 0.
        \] 
If \( h_i \) is not an NCBF, then for the boundary condition \( x_i = \bar{x}_i \), \( h_i(x_i) = 0 \), and for all \( u_i \), the control condition is not met:
        \[
        u_i < \left(1 - \bar{x}_i - r_i\right)\left(\beta_{ii} + \sum_{j \in \mathcal{N}_i^+} \beta_{ij} \frac{\bar{x}_j}{\bar{x}_i}\right)  -\gamma_i,
        \]
        which implies: \[
        \dot{h}_i(\mathbf{x}_i, r_i, u_i) < - \alpha_i(\bar{x}_i - x_i).
        \]
        Therefore, \( h_i(x_i(t)) \) would decrease below zero, causing \( x_i(t) > \bar{x}_i \), which violates forward invariance. Thus,
    \(h_i\) must be an NCBF if \(S_i \subseteq C_i\) is forward invariant.
    \hfill \qedhere
\end{proof}
This result guarantees that, whenever a node has a valid NCBF, its safety set remains forward invariant. Theorem~\ref{thm:forwardInvariance} implies that, for \(h_i(x_i)\) to be an NCBF, a feasible control input \(u_i \in U_i\) must exist for all states within \(S_i\). We next derive a consequence relating to the necessary control input.
\begin{corollary}\label{cor:feas}
If the maximum control effort \(\bar{u}_i\) satisfies
\[
\bar{u}_i \geq (1 - \bar{x}_i - r_i^0) \beta_{ii} + (1 -  \bar{x}_i - r_i^0)\sum_{j=1}^n \beta_{ij}\frac{\bar{x}_j}{\bar{x}_i} - \gamma_i,
\]
where \( r_i^0 = r_i(0) \), there will always exist an admissible control input \(u_i \in U_i\) that guarantees node $i$ is safe.
\end{corollary}
By Definition~\ref{dfn:safe} and the proof of Theorem~\ref{thm:forwardInvariance}, Corollary~\ref{cor:feas} considers the worst-case conditions for node~$i$'s safety where all nodes are at their infection threshold and the recovered state is as small as possible. If the initial condition of the recovered state is not known, $r_i^0$ can be replaced by zero, assuming those that are not currently infected are susceptible. 

Since real epidemics involve model and reporting uncertainty, we extend the framework with robust CBFs that add an uncertainty buffer to preserve these guarantees.

\section{Robust CBFs for the Networked SIR Model} \label{sec: robust}
In this section, we include uncertainty in our framework. We first implement the robustness framework developed in \cite{alan2025generalizing} on the networked SIR model. We begin with the general uncertain system
\begin{equation}
\label{dynamicuncertain}
\begin{pmatrix}\dot{x}_i \\ \dot{r}_i \end{pmatrix} = f_i(\mathbf{x}_i,r_i) + g_i(x_i) u_i + \mu_i(\mathbf{x}_i,u_i),   
\end{equation}
where, for each node \(i\), the uncertainty \(\mu_i(\mathbf{x}_i,u_i)\)
captures the unmodeled effects or disturbances. 
When uncertainty is present, evaluating the time derivative of \(h_i(x_i)\), with respect to the full dynamics in \eqref{dynamicuncertain}, yields
\[
\dot{h}_i(\mathbf{x}_i, r_i, u_i) =  \mathcal{L}_{f_i} h_i(\mathbf{x}_i) + \mathcal{L}_{g_i} h_i(x_i) u_i + 
\mathcal{L}_{\mu_i} h_i(\mathbf{x}_i, u_i),
\]
where
\[
\mathcal{L}_{\mu_i} h_i(\mathbf{x}_i, u_i) = \frac{\partial h_i}{\partial x_i}\mu_i(\mathbf{x}_i,u_i).
\]
Since \(h_i(x_i) = \bar{x}_i - x_i\) and \(\frac{\partial h_i}{\partial x_i}=-1\), the term \(\mathcal{L}_{\mu_i} h_i(\mathbf{x}_i, u_i)\) simplifies to
\[
\mathcal{L}_{\mu_i} h_i(\mathbf{x}_i, u_i) = -\mu_i(\mathbf{x}_i,u_i).
\]

The uncertainty prevents us from finding a control input \(u_i\) that satisfies \eqref{safety} for all possible \(\mu_i\). The robust CBF (RCBF) framework \cite{alan2025generalizing} addresses this issue by introducing a compensation term \(\sigma_i(\mathbf{x}_i,u_i)\). 
From Lemma 2 of \cite{alan2025generalizing}, a sufficient condition to guarantee robust safety, is that the compensation term satisfies
\begin{equation}
\label{lie_comp}
\mathcal{L}_{\mu_i} h_i(\mathbf{x}_i, u_i) + \sigma_i(\mathbf{x}_i,u_i) \ge 0.    
\end{equation}
We aim to find a set of controllers \(u_i \in \mathcal{K}^{RCBF}_i \subseteq U_i\) where 
\begin{align}
\label{generalrobustset}
 \mathcal{K}^{RCBF}_i = \Bigl\{ u_i \in U_i \;:\;\mathcal{L}_{f_i} h_i(\mathbf{x}_i) + \mathcal{L}_{g_i} h_i(x_i) u_i \nonumber\\
 - \sigma_i(\mathbf{x}_i,u_i) \ge -\alpha_i\bigl(h_i(x_i)\bigr), \quad \text{and} \nonumber\\
 \mathcal{L}_{\mu_i} h_i(\mathbf{x}_i, u_i) + \sigma_i(\mathbf{x}_i,u_i) \ge 0 \Bigr\} .  
\end{align}
In other words, any controller \(u_i\) chosen from the set \(\mathcal{K}^{RCBF}_i\) guarantees that the modified robust safety condition
\begin{equation}
\label{robustcondition}
\mathcal{L}_{f_i} h_i(\mathbf{x}_i) + \mathcal{L}_{g_i} h_i(x_i) u_i - \sigma_i(\mathbf{x}_i,u_i) \ge -\alpha_i\bigl(h_i(x_i)\bigr)    
\end{equation}
holds, while the inequality in \eqref{lie_comp} must also be satisfied.
\noindent
We now adapt the notion of robust safety from \cite{alan2025generalizing} to the epidemic setting.
\begin{definition}
\label{robust_safety}
A node \(i\) is \emph{robustly safe} if for every admissible realization of the uncertainty \(\mu_i(\mathbf{x}_i,u_i)\) and for every initial state \(x_i(0) \in C_i\), there exists a control input \(u_i \in U_i\) such that \eqref{robustcondition} holds.
\end{definition}
Uncertainty plays a central role in epidemic processes as case counts are often under-reported, and small outbreaks are difficult to measure accurately. As a result, nominal safety guarantees may not hold unless the uncertainty is addressed. To address this challenge, we design compensation terms within the RCBF framework that capture how uncertainty may appear in epidemic dynamics. We describe two compensation term~\((\sigma_i)\) choices.

\subsection{Uncertainty Bounded Independently of Control Input}
In a networked SIR model, each node may face reporting errors or random fluctuations, such as unpredictable variations in disease transmission due to behavioral changes or environmental factors, that do not depend on how aggressively a node reduces infection. For instance, a constant fraction of under-reported infections or seasonal variations might persist at each node irrespective of local quarantine rules. In this case, the uncertainty bound \(\bar{\mu}_i\) can be treated as a fixed offset in the CBF constraint, which leads to a more conservative but straightforward policy: every node’s controller always reserves enough of a safety margin to handle the worst-case offset. Although it may overestimate errors when infection levels are small, the independent approach ensures that each node remains robustly safe across varying network interactions without needing detailed assumptions on how control influences uncertainty.

Assume that the uncertainty is bounded above and there exists an upper bound \(\bar{\mu}_i \geq 0\) which satisfies
\begin{equation}
\label{func_independent}
\|\mu_i(\mathbf{x}_i,u_i)\| \le \bar{\mu}_i \quad \forall{ \mathbf{x}_i,u_i}.    
\end{equation}
A natural choice is to assume \(\sigma_i(\mathbf{x}_i) = \bar{\mu}_i\,\|\frac{\partial h_i}{\partial \mathbf{x}_i}\|\).
Since \(\|\frac{\partial h_i}{\partial x_i}\|=1\), we obtain 
\begin{equation}
\label{comp1}
 \sigma_i = \bar{\mu}_i.  
\end{equation}

\subsection{Low-Prevalence--Amplified Uncertainty} \label{subsec:lowprev}
Early in an outbreak (or after successful suppression), the infected fraction \(x_i\) at a node can be extremely small. In this regime, case counts are drawn from very limited samples and are prone to under-reporting, delayed confirmation, and high relative statistical noise. Surveillance systems are known to fail precisely in these low-prevalence scenarios, where algorithm-based detection performs poorly when there are a limited number of cases.

If reported infections follow Poisson counting statistics, the standard deviation of the count grows with the square root of the count. Under the Poisson assumption, the reported infection fraction~\((\hat{x}_i)\) has \(\mathrm{std}(\hat{x}_i)\propto \sqrt{\hat{x}_i}\) (absolute error). However, the coefficient of variation \(\mathrm{std}(\hat{x}_i)/\mathbb{E}[\hat{x}_i]\propto 1/\sqrt{\hat{x}_i}\) (where std denotes the standard deviation and $\mathbb{E}$ denotes the expectation), so the relative error explodes as \(\hat{x}_i \to 0\) \cite{Poisson}. This breakdown of Poisson approximations at low rates has been documented in disease surveillance models using Poisson-Kalman filtering, where the approximation breaks down when the rates of occurrences are much smaller \cite{Baike}.

To capture this phenomenon, we introduce an uncertainty whose magnitude increases as \(x_i \to 0\):
\begin{equation}
\label{func_lowprev}
\mu_i(\mathbf{x}_i,u_i)
= \frac{d_i(\mathbf{x}_i,u_i)}{\sqrt{(x_i+\delta)}}, 
\end{equation}
where \(0<\delta\ll1\) regularizes the expression near \(x_i=0\) and \(d_i(\mathbf{x}_i,u_i)\) represents the underlying uncertainty component. We assume \(\|d_i(\mathbf{x}_i,u_i)\|\le\bar d_i\). Consequently, the Lie derivative computation becomes
\[
\mathcal{L}_{\mu_i} h_i(\mathbf{x}_i,u_i)
= \frac{\partial h_i}{\partial x_i}\,\mu_i(\mathbf{x}_i,u_i)
= -\frac{d_i(\mathbf{x}_i,u_i)}{\sqrt{(x_i+\delta)}},
\]
so that in the worst-case
\[
\|\mathcal{L}_{\mu_i} h_i(\mathbf{x}_i,u_i)\|
\le \frac{\bar d_i}{\sqrt{(x_i+\delta)}},
\]
leading to the conservative compensation term
\begin{equation}
\label{ncomp}
\sigma_i({x}_i)
= \frac{\bar d_i}{\sqrt{(x_i+\delta)}}.
\end{equation}

\subsection{Robust Safety Synthesis and Theoretical Guarantees}
Depending on the type of uncertainty, the compensation terms \eqref{comp1} or \eqref{ncomp} can be chosen for node \(i\).
A controller \(u_i\) chosen from the robust safe set
\begin{align}
\label{robustset}
 \mathcal{K}^{RCBF}_i = \Bigl\{ u_i \in U_i \;&:\; \gamma_i x_i - \left(1-x_i-r_i\right)\sum_{j=1}^n \beta_{ij} x_j \nonumber\\
 &+ x_i\,u_i - \sigma_i(\mathbf{x}_i,u_i) \ge -\alpha_i(\bar{x}_i-x_i) \Bigr\} ,  
\end{align}
where the requirement, from \eqref{lie_comp},
\[
-\mu_i(\mathbf{x}_i,u_i) + \sigma_i(\mathbf{x}_i, u_i) \ge 0,
\]
ensures that the safety condition holds even in the worst-case realization of the uncertainty.

\begin{assumption}
\label{exitsigma}
Let there exist a compensation term \(\sigma_i(\mathbf{x}_i,u_i)\) such that the modified robust CBF condition in \eqref{robustcondition} holds, for all \((\mathbf{x}_i,u_i)\), and also satisfies \eqref{lie_comp}.  \end{assumption}
Under this assumption on the compensation term, we can state a robust counterpart to our earlier forward invariance result, ensuring safety despite unmodeled disturbances.
\begin{lemma}[Forward Invariance with Robust CBF]
\label{lem:RCBF_forward_invariance}
Consider the uncertain dynamics for node \(i\) in \eqref{dynamicuncertain}, and define the safety constraint \(h_i(x_i)\) in \eqref{safety} with the safe set \(C_i\) from \eqref{safeset}.
Under Assumption \ref{exitsigma}, for any control input \(u_i\) satisfying \eqref{robustcondition}, the safe set \(C_i\) is forward invariant.
\end{lemma}

\begin{proof}
Since \(h_i(x_i)=\bar{x}_i-x_i\), the Lie derivative of \(h_i\) with respect to the uncertainty vector field \(\mu_i\) is 
\[
\mathcal{L}_{\mu_i} h_i(\mathbf{x}_i,u_i) = \frac{\partial h_i}{\partial x_i} \mu_i(\mathbf{x}_i,u_i) = -\mu_i(\mathbf{x}_i,u_i).
\]
Thus, \eqref{lie_comp} can be written as
\[
-\mu_i(\mathbf{x}_i,u_i) + \sigma_i(\mathbf{x}_i,u_i) \ge 0,
\]
which implies that the worst-case effect of the uncertainty is compensated by \(\sigma_i\).
Recalling the time derivative of \(h_i(x_i)\),
\[
\dot{h}_i(\mathbf{x}_i, r_i, u_i) = \mathcal{L}_{f_i} h_i(\mathbf{x}_i, r_i) + \mathcal{L}_{g_i} h_i(x_i) u_i + \mathcal{L}_{\mu_i} h_i(\mathbf{x}_i,u_i).
\]
By adding and subtracting \(\sigma_i(\mathbf{x}_i,u_i)\), we obtain
\begin{align}
\dot{h}_i(\mathbf{x}_i, r_i, u_i) &= \Bigl[\mathcal{L}_{f_i} h_i(\mathbf{x}_i, r_i) + \mathcal{L}_{g_i} h_i(x_i) u_i - \sigma_i(\mathbf{x}_i,u_i)\Bigr] \nonumber\\ 
&+ \Bigl[\mathcal{L}_{\mu_i} h_i(\mathbf{x}_i,u_i) + \sigma_i(\mathbf{x}_i,u_i)\Bigr]. \nonumber    
\end{align}
By \eqref{robustcondition}, the terms in the first bracket are bounded below by \(-\alpha_i\bigl(h_i(x_i)\bigr)\) and, by \eqref{lie_comp}, the second bracket is nonnegative. Hence,
\[
\dot{h}_i(\mathbf{x}_i, r_i, u_i) \ge -\alpha_i\bigl(h_i(x_i)\bigr).
\]
By Theorem 1 (Nagumo's theorem) of \cite{alan2025generalizing}, it is well-known that if \(\dot{h}_i(\mathbf{x}_i, r_i, u_i) \ge -\alpha_i\bigl(h_i(x_i)\bigr)\) for an extended class-\( \mathcal{K} \) function \(\alpha_i\), then the set \(C_i = \{ x_i \in [0,1] : h_i(x_i) \geq 0 \}\) is forward invariant.
\hfill \qedhere
\end{proof}
To formally guarantee the existence of a robustly safe controller, we impose the following assumption and subsequently provide the corresponding theoretical result.
\begin{assumption} 
\label{assumption:RCBF_existence}
For node \(i\) with dynamics defined in \eqref{dynamicuncertain} and safety function \(h_i(x_i)\) defined in \eqref{safety}, assume that there exists a compensation term \(\sigma_i(\mathbf{x}_i,u_i)\) selected according to \eqref{comp1} or \eqref{ncomp}, and that for every \(x_i\) in the viability domain, the set \(\mathcal{K}^{RCBF}_i\), defined in \eqref{robustset}, is non-empty. Furthermore, there exists a class-\(\mathcal{K}\) function \(\alpha\) such that the set \(\mathcal{K}^{RCBF}_i\) remains non-empty for all \(x_i \in (0, \bar{x}_i]\).
\end{assumption}

\begin{proposition}[Existence of Robust Safe Controller]
\label{thm:existence_RCBF_controller}
Under Assumption~\ref{assumption:RCBF_existence}, for any input \(u^*_i \in \mathcal{K}^{RCBF}_i\), the closed-loop dynamics remain robustly safe, meaning that the safe set \(C_i\) defined in \eqref{safeset} is forward invariant.
\end{proposition}

\begin{proof}
By Assumption~\ref{assumption:RCBF_existence}, for any \(x_i\), the set \(\mathcal{K}^{RCBF}_i\), defined in \eqref{robustset}, is non-empty. Thus, for each state \(x_i\), there exists at least one control input \(u^*_i\) which satisfies the RCBF inequality:
\[
\mathcal{L}_{f_i} h_i(\mathbf{x}_i) + \mathcal{L}_{g_i} h_i(x_i)u^*_i - \sigma_i(\mathbf{x}_i,u^*_i) \ge -\alpha_i\bigl(h_i(x_i)\bigr).
\]
In addition, by construction, the compensation term satisfies
\[
\mathcal{L}_{\mu_i} h_i(\mathbf{x}_i,u^*_i) + \sigma_i(\mathbf{x}_i,u^*_i) \ge 0,
\]
then the total time derivative of \(h_i(x_i)\) along the closed-loop trajectories is
\[
\dot{h}_i(\mathbf{x}_i, r_i, u_i) \ge -\alpha_i\bigl(h_i(x_i)\bigr).
\]
It follows that the robust safety condition required for the forward invariance of the safe set \(C_i\) is satisfied by the closed-loop dynamics under the controller \(u_i^*\). By directly invoking Lemma~\ref{lem:RCBF_forward_invariance}, we conclude that the set \(C_i\) is forward invariant, implying that the chosen controller \(u_i^*\) renders the system robustly safe.
\hfill \qedhere
\end{proof}

\begin{corollary}\label{feas_rcbf}
Suppose Assumption~\ref{exitsigma} holds.
If the maximum control effort~$\bar u_i$ satisfies
\begin{equation*}\label{rcbf_bound}
\bar u_i \;\ge\; 
(1 - \bar{x}_i - r_i^0) \beta_{ii} 
+ (1 -  \bar{x}_i - r_i^0)\sum_{j=1}^n \beta_{ij}\frac{\bar{x}_j}{\bar{x}_i} 
- \gamma_i
+ \frac{\bar{\sigma_i}}{\bar{x}_i},
\end{equation*}
where \( r_i^0 = r_i(0) \) and $\bar{\sigma_i} = \max_{\mathbf{x}_i}\sigma_i(\mathbf{x}_i,u_i)$,
there will always exist an admissible input $u_i\in U_i=[0,\bar u_i]$ that guarantees node $i$ is robustly safe.
\end{corollary}


\section{Experiments and Results}\label{secres}
\begin{figure*}[t]
    \centering
    \begin{subfigure}[t]{0.32\linewidth}
        \centering
        \includegraphics[width=\linewidth]{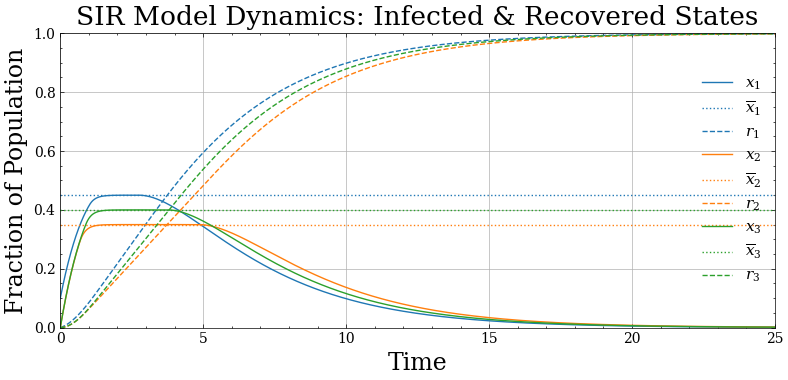}
        \caption{Nominal dynamics}
        \label{fig:state_dynamics}
    \end{subfigure}
    \hfill
    \begin{subfigure}[t]{0.32\linewidth}
        \centering
        \includegraphics[width=\linewidth]{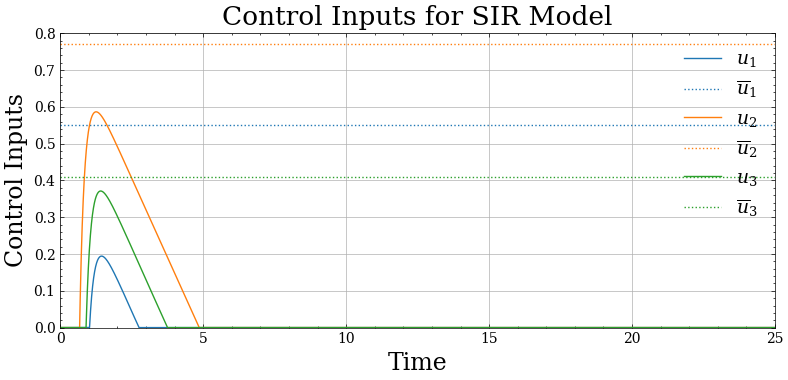}
        \caption{Control inputs under nominal conditions}
        \label{fig:control_inputs}
    \end{subfigure}
    \hfill
    \begin{subfigure}[t]{0.32\linewidth}
        \centering
        \includegraphics[width=\linewidth]{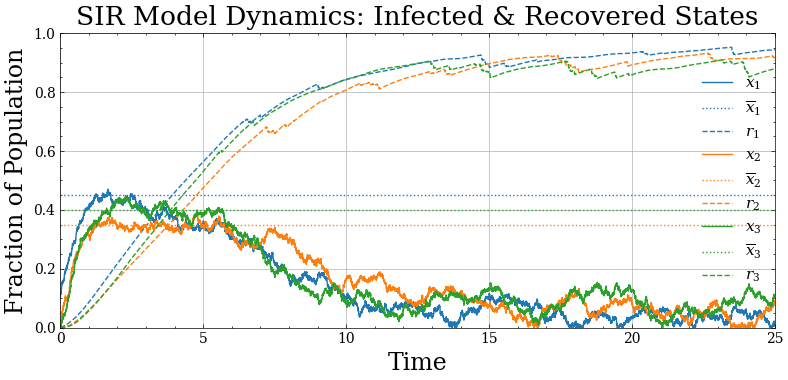}
        \caption{Noisy dynamics}
        \label{fig:state_dynamics_noise}
    \end{subfigure}
    \caption{Time evolution of the infected and recovered states in the networked SIR model under nominal conditions and with noise, along with the corresponding control inputs.}
    \label{fig:overall_three_figures}
\end{figure*}

We evaluate the proposed safety-critical control framework on a networked SIR model with \(n=3\) interconnected nodes. The dynamics at each node are governed by \eqref{eq:SIR}.
In our simulation, the network topology is defined by a set of directed edges such that each node influences every other node. 
The infection rate matrix \((B)\) is defined as 
\[
B = 0.55\,I_n + 0.45\,(1_n1_n^\top - I_n),
\]
which implies \(\beta_{ii}=0.55\) and \(\beta_{ij}=0.45\) for \(i\neq j\). The recovery rate is set uniformly as \(\gamma_i=0.3\) for all nodes. In our setup, the maximum acceptable infection thresholds \(\bar{x}_i = [0.45, 0.35, 0.4]\) while the maximum control inputs \(\bar{u}_i = [0.55, 0.77, 0.41]\), which are random numbers between \(0\) and \(1\).

The simulation is conducted over the time interval \(t\in[0,T]\) with a time step of \(dt=10^{-4}\) and $T=25$. The initial conditions are chosen so that node \(i=1\) starts with an infection level of \(x_1(0) = 1-s_1(0)=0.1\) (with \(s_1(0)=0.9\)) while the other nodes begin with zero infection. At each time step, a safety-critical control input \(u_i\) is computed at every node in accordance with the CBF condition~\eqref{finalu}, to ensure that the infection level remains below the threshold~\(\bar{x}_i\).

\subsection{CBF under Nominal Conditions}
We simulate the CBF where the safety-critical controller is computed based on the barrier function condition \eqref{safety} and ensures that the infection level \(x_i\) remains below the threshold \(\bar{x}_i\) in the absence of uncertainty.
When a node's infection level begins to approach or exceed its safety threshold, the corresponding control input increases within its allowable range to enhance the healing rate and push the state back into the safe region. 
All nodes' infected and recovered states are depicted in Figure~\ref{fig:state_dynamics}, where the infection levels are kept below their safety criteria. The control inputs applied to each node over time are shown in Figure~\ref{fig:control_inputs}, which remain within the prescribed limits.
The simulation results (Figures~\ref{fig:state_dynamics} and \ref{fig:control_inputs}) demonstrate that the nominal CBF controller effectively limits the infection levels below their safety thresholds while activating control inputs only when needed. These observations indicate that, in the absence of disturbances, the controller efficiently balances the trade-off between minimal intervention and epidemic containment.

We also examine the CBF under the disturbance condition by applying i.i.d. Gaussian process noise with a variance of \(1.0\). Figure~\ref{fig:state_dynamics_noise} shows that the fractions of infected and recovered individuals in all nodes fluctuate within their allowable ranges. It also demonstrates that occasionally, node infections exceed their thresholds due to the presence of noise. These results indicate that the CBF is not reliable under uncertainty, particularly if the noise scale increases or the maximum safe infection threshold decreases. To address this issue, we apply robust CBF methods to our noisy network.

\begin{figure*}[t]
    \centering
    \begin{subfigure}[t]{0.48\linewidth}
        \centering
        \includegraphics[width=\linewidth]{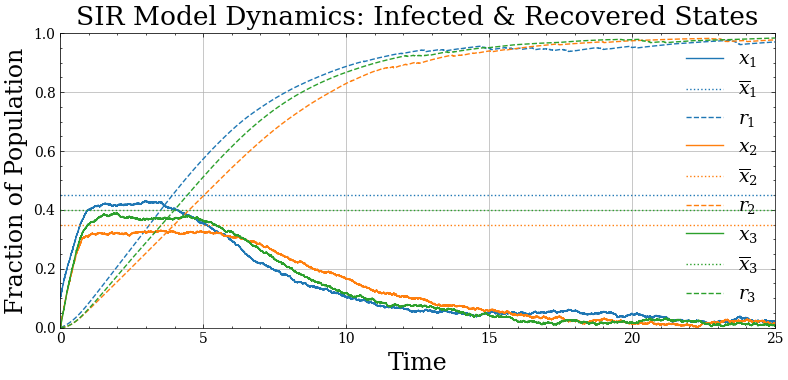}
        \caption{Dynamics with Independent Uncertainty}
        \label{fig:state_dynamics_independent_noise}
    \end{subfigure}
    \hfill
    \begin{subfigure}[t]{0.48\linewidth}
        \centering
        \includegraphics[width=\linewidth]{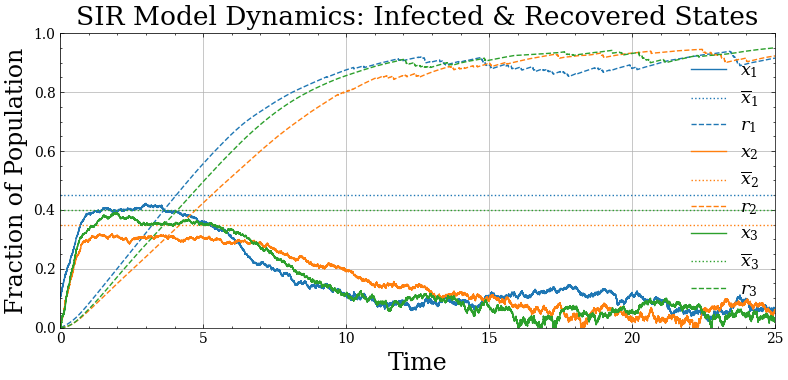}
        \caption{Dynamics with Low-Prevalence--Amplified Uncertainty}
        \label{fig:lowprev_noise}
    \end{subfigure}
    \hfill
    \vspace{1em}
    \begin{subfigure}[t]{0.48\linewidth}
        \centering
        \includegraphics[width=\linewidth]{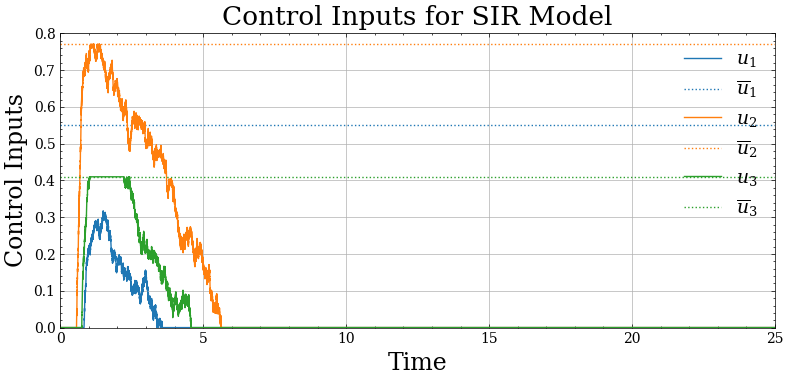}
        \caption{Control input (Independent Uncertainty)}
        \label{fig:input_independent_noise}
    \end{subfigure}
    \hfill
    \begin{subfigure}[t]{0.48\linewidth}
        \centering
        \includegraphics[width=\linewidth]{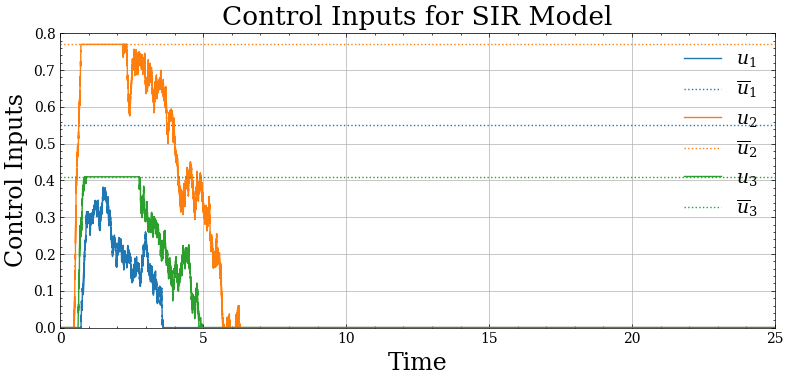}
        \caption{Control input (Low-Prevalence--Amplified Uncertainty)}
        \label{fig:state_input_low_noise}
    \end{subfigure}
    \hfill
    
    \caption{Time evolution of the infected and recovered states and control inputs in the noisy networked SIR model under independent and low-prevalence--amplified uncertainty.}
    \label{fig:overall_noise}
\end{figure*}

\subsection{RCBF under Uncertainties}
We implement the robust CBF framework by robustifying the barrier condition through a compensation term \(\sigma_i(\mathbf{x}_i,u_i)\) as described in Section~\ref{sec: robust}. Two choices for the compensation term are considered: \textit{independent}~\eqref{comp1} and \textit{low-prevalence--amplified}~\eqref{ncomp}.

To evaluate our robust framework, we examine how each compensation term choice between \eqref{comp1} and \eqref{ncomp} performs under various noise structures reflecting different uncertainty models. Specifically, we consider: (i) \emph{independent} uncertainty, modeled by additive Gaussian noise $\mathcal{N}(0,1)$, which we bound by setting the upper uncertainty level to $\bar{\mu}_i = 0.15$, independent of both \(x_i\) and \(u_i\); (ii) \emph{low-prevalence--amplified} uncertainty, which scales with the infection level so that \(\mu_i(\mathbf{x}_i,u_i) = \frac{\epsilon_i}{\sqrt{(x_i+\delta)}}\), where \(\epsilon_i \sim \mathcal{N}(0, 1)\), $\delta = 0.01$, and the bound on the disturbance term is set to $\bar{d}_i = 0.15$. 

\begin{table}[b]
    \centering
    \caption{Maximum infection and control input for each node under two approaches.}
    \label{tab:peakvals}
    \begin{tabular}{c|cc|cc}
      \hline
      \multirow{2}{*}{\textbf{Node}} 
      & \multicolumn{2}{c|}{\textbf{Independent}}
      & \multicolumn{2}{c}{\textbf{Low-Prevalence--Amplified}} \\
      \cline{2-5}
      & $x^{\max}_i$ & $u^{\max}_i$ & $x^{\max}_i$ & $u^{\max}_i$\\
      \hline
      Node 1 & 0.426 & 0.316 & 0.422 & 0.382 \\
      Node 2 & 0.329 & 0.770 & 0.319 & 0.770 \\
      Node 3 & 0.396 & 0.410 & 0.391 & 0.410 \\
      \hline
    \end{tabular}
\end{table}

\begin{table}[b!]
    \centering
    \caption{Aggregate results across the entire network.}
    \label{tab:aggregatevals}
    \begin{tabular}{l|c|c}
      \hline
      \textbf{Method} & \textbf{Avg. Min Safety Margin} & \textbf{Integrated Control}\\
      \hline
      Independent     & 0.016 & 3.311 \\
      \shortstack{Low-Prevalence} & 0.023 & 4.858 \\
      \hline
    \end{tabular}
\end{table}

Figure~\ref{fig:overall_noise} compares the independent and low-prevalence--amplified compensation strategies applied to the noisy networked SIR model. As shown in Figures~\ref{fig:state_dynamics_independent_noise} and \ref{fig:input_independent_noise} for the independent approach, and Figures~\ref{fig:lowprev_noise} and \ref{fig:state_input_low_noise} for the low-prevalence--amplified approach, both RCBF controllers effectively maintain all infection trajectories below their respective thresholds, though their operational characteristics differ. The independent strategy assumes a constant worst-case disturbance bound~\eqref{func_independent}, resulting in relatively consistent control interventions throughout the simulation horizon. In contrast, the low-prevalence formulation~\eqref{ncomp} allocates greater compensation when infection levels are low, resulting in stronger early actions that suppress incipient outbreaks and create additional margin against noise, but at the expense of greater overall input effort. Additional distinctions are evident in the dynamics: for instance, the low-prevalence--amplified method exhibits lower peak infection levels across nodes but incurs higher cumulative control effort. From these comparisons, we learn that the choice between strategies involves a fundamental trade-off between control efficiency (favoring the independent approach in scenarios with stable disturbance levels) and enhanced conservatism (favoring the low-prevalence--amplified approach in high-uncertainty, low-infection regimes), allowing decision-makers to select based on the specific epidemic context and resource constraints.

Table~\ref{tab:peakvals} summarizes the highest infection levels $(\max_{t \geq 0} x_i(t))$ and the maximum control actions \((\max_{t \geq 0} u_i(t))\) at each node under the two approaches. Table~\ref{tab:aggregatevals} aggregates the results across the network. The `Avg. Min Safety Margin' column shows the average of the minimum difference between the infection threshold and the infection level over time, across all nodes, computed as $\frac{1}{n} \sum_{i=1}^n \min_{t \geq 0} (\bar{x}_i - x_i(t))$; higher values indicate greater conservatism in maintaining safety buffers. The `Integrated Control' column quantifies the total control effort across the network as $\int_0^T \sum_{i=1}^n u_i(t) \, dt$. These tables confirm the trade-off observed in Figure~\ref{fig:overall_noise}: the independent strategy uses less cumulative effort but operates closer to the safety boundary, while the low-prevalence strategy uses more resources to guarantee a larger safety cushion.

\section{Conclusion and Future Work} \label{secfin}
We presented a safety-critical control framework for a networked SIR epidemic model using CBFs and RCBFs. By treating infection-level constraints as barrier functions and incorporating noise through a compensation term \((\sigma_i)\), our approach offers a systematic method to keep infections below preset thresholds despite measurement noise and model uncertainties. Simulation results demonstrate that our novel low-prevalence--amplified formulation is more conservative than the independent approach, allocating stronger interventions when infection levels are small to maintain a larger safety margin, while the independent approach is less conservative and therefore more control-efficient, keeping trajectories safe but operating closer to the prescribed thresholds. These robust methods are essential under higher uncertainty or structured disturbances, where nominal CBFs fail to provide safety guarantees. 
Future work will investigate realistic and time-varying disturbance models, including reporting delays.

	\bibliographystyle{IEEEtran}
	\bibliography{references}
\end{document}